\documentclass[10pt,onecolumn,twoside]{article}

\usepackage{amsmath}
\usepackage{graphicx}
\usepackage{amssymb}
\usepackage{bm}
\usepackage{mathtools}
\usepackage{cite}
\usepackage{amsthm}
\usepackage{hyperref}
\usepackage{float}
\usepackage{url, doi}
\usepackage[noend]{algpseudocode}
\usepackage{algorithm}
\usepackage{fullpage}

\usepackage{booktabs}
\usepackage{multirow}       
\usepackage{tabularx}       
\usepackage{hhline}
\usepackage{arydshln}		
\usepackage{xcolor}
\usepackage{soul}
\usepackage[colorinlistoftodos, prependcaption, textsize=tiny]{todonotes}

\algtext*{EndFor}

\definecolor{lightgreen}{rgb}{.9,1,.9}
\newcolumntype{L}[1]{>{\raggedright\arraybackslash}p{#1}}
\newcolumntype{C}[1]{>{\centering\arraybackslash}p{#1}}
\newcolumntype{R}[1]{>{\raggedleft\arraybackslash}p{#1}}
\newcolumntype{M}[1]{>{\centering\arraybackslash}m{#1}}

\DeclarePairedDelimiterX{\normsz}[1]{\lVert}{\rVert}{#1}

\makeatletter
\if@todonotes@disabled

\else

\fi
\makeatother

\def\DnCNNast{{\text{DnCNN}^\ast}}

\def\defn{\,\coloneqq\,}

\def\argmin{\mathop{\mathsf{arg\,min}}} 
\def\prox{\mathsf{prox}}

\def\R{\mathbb{R}}
\def\N{\mathbb{N}}

\def\ebm{{\bm{e}}}
\def\xbm{{\bm{x}}}
\def\ybm{{\bm{y}}}
\def\ubm{{\bm{u}}}

\def\nbm{{\bm{n}}}
\def\sbm{{\bm{s}}}

\def\zbm{{\bm{z}}}
\def\zerobm{{\bm{0}}}

\def\d{{\, \mathrm{d}}}
\def\e{{\mathrm{e}}}

\def\Dsf{{\mathsf{D}}}
\def\Fsf{{\mathsf{F}}}

\def\xbmhat{{\widehat{\bm{x}}}}

\def\Ibm{{\bm{I}}}
\def\Hbm{{\bm{H}}}
\def\Dbm{{\bm{D}}}

\def\Ncal{\mathcal{N}}

\theoremstyle{definition}
\newtheorem{proposition}{Proposition}

\definecolor{pink}{HTML}{EB346F}
\definecolor{lightgreen}{RGB}{229,251,229}

\hypersetup{
    colorlinks=true,       
    linkcolor=pink,          
    citecolor=pink,        
    urlcolor=pink           
}

\title{Boosting the Performance of Plug-and-Play Priors via Denoiser Scaling} 

\author{Xiaojian~Xu%
\thanks{Department of Computer Science \& Engineering, Washington University in St.~Louis, St.~Louis, MO 63130.}
\hspace{0.05em},
Jiaming~Liu%
\thanks{Department of Electrical \& Systems Engineering, Washington University in St.~Louis, St.~Louis, MO 63130.}
\hspace{0.05em}, 
Yu~Sun$^\ast$,\\
Brendt~Wohlberg%
\thanks{Theoretical Division, Los Alamos National Laboratory, Los Alamos, NM 87545 USA.}
\hspace{0.05em}, 
and Ulugbek~S.~Kamilov$^{\ast, \dagger}$}

\begin{document}

\maketitle


\begin{abstract}
\emph{Plug-and-play priors (PnP)} is an image reconstruction framework that uses an image denoiser as an imaging prior. Unlike traditional regularized inversion, PnP does not require the prior to be expressible in the form of a regularization function. This flexibility enables PnP algorithms to exploit the most effective image denoisers, leading to their state-of-the-art performance in various imaging tasks. In this paper, we propose a new \emph{denoiser scaling} technique to explicitly control the amount of PnP regularization. Traditionally, the performance of PnP algorithms is controlled via intrinsic parameters of the denoiser related to the noise variance. However, many powerful denoisers, such as the ones based on convolutional neural networks (CNNs), do not have tunable parameters that would allow controlling their influence within PnP. To address this issue, we introduce a scaling parameter that adjusts the magnitude of the denoiser input and output. We theoretical justify the denoiser scaling from the perspectives of proximal optimization, statistical estimation, and consensus equilibrium. Finally, we provide numerical experiments demonstrating the ability of denoiser scaling to systematically improve the performance of PnP for denoising CNN priors that do not have explicitly tunable parameters.
\end{abstract}


\section{Introduction}
\label{Sec:Intro}

Image formation is naturally posed as an inverse problem, which is often ill-posed. Regularized inversion is a widely adopted framework for dealing with such ill-posed inverse problems by taking advantage of prior information regarding the unknown image. Examples of popular imaging priors include nonnegativity, transform-domain sparsity, and self-similarity~\cite{Rudin.etal1992, Figueiredo.Nowak2001, Elad.Aharon2006, Danielyan.etal2012, Yang.Jacob2013, Kamilov2017}. Since imaging priors are often nondifferentiable, proximal algorithms~\cite{Parikh.Boyd2014} -- such as variants of iterative shrinkage/thresholding algorithm (ISTA)~\cite{Figueiredo.Nowak2003, Daubechies.etal2004, Bect.etal2004, Beck.Teboulle2009a} and alternating direction method of multipliers (ADMM)~\cite{Eckstein.Bertsekas1992, Afonso.etal2010, Ng.etal2010, Boyd.etal2011} -- are extensively used in image reconstruction. These algorithms avoid differentiating the regularizer by using the \emph{proximal operator}~\cite{Parikh.Boyd2014}, a mathematical concept that corresponds to regularized image denoising.

\emph{Plug-and-play priors (PnP)}~\cite{Venkatakrishnan.etal2013} is a framework for regularized inversion that does not require the image prior to be expressible in the form of a regularization function. Recent results have shown that by using advanced image denoisers in iterative image reconstruction, PnP algorithms achieve state-of-the-art performance in many imaging problems~\cite{Chan.etal2016, Sreehari.etal2016, Ono2017, Kamilov.etal2017, Zhang.etal2017a, Buzzard.etal2017, Teodoro.etal2019}. PnP algorithms have also been successfully combined with powerful denoising \emph{convolutional neural nets (CNNs)} for exploiting learned imaging priors while enforcing fidelity to the measured data~\cite{Sun.etal2018a, Tirer.Giryes2019, Meinhardt.etal2017, Ryu.etal2019, Song.etal2020}. (A similar strategy for using general image denoisers has also been adopted in related frameworks known as \emph{approximate message passing (AMP)}~\cite{Tan.etal2015, Metzler.etal2016, Metzler.etal2016a, Fletcher.etal2018} and \emph{regularization by denoising (RED)}~\cite{Romano.etal2017, Bigdeli.etal2017, Metzler.etal2018, Reehorst.Schniter2019, Mataev.etal2019, Sun.etal2019b}.)

The price of the convenience of using a general denoiser is that PnP algorithms lose interpretability as optimization methods for minimizing explicit convex objective functions. This complicates their theoretical analysis, a problem that was extensively addressed in several recent publications~\cite{Chan.etal2016, Sreehari.etal2016, Buzzard.etal2017, Sun.etal2018a, Ryu.etal2019}.  A distinct but related problem is the lack of a \emph{regularization parameter} for adjusting the relative strength between the prior and the data fidelity.
Unlike traditional regularized inversion, PnP does not have an \emph{explicit} relationship between the regularization parameter and the intrinsic parameters of the denoiser.
Instead, current PnP algorithms rely on indirect proxies for adjusting the relative strength of the prior, such as the noise-variance parameter (available for some denoisers) or the parameters of the iterative algorithms. However, many denoisers, such as those based on CNNs, do no have tunable parameters controlling their influence within PnP. This leads to their suboptimal performance, and requires training several CNN instances at multiple noise levels.

To address this issue, we introduce a tunable regularization parameter for PnP, which is independent from the intrinsic parameters of the denoiser or iterative algorithms. More specifically, the key contributions of this paper are as follows:
\begin{itemize}
\item We introduce a new \emph{denoiser scaling} technique that simply scales the denoiser input by a positive constant and its output by the inverse of the same constant. The technique is broadly applicable to all PnP algorithms, and provides a mechanism to adjust the denoiser strength in a way that is independent of traditional approaches.

\item We present a detailed theoretical justification of denoiser scaling for several classes of denoisers. We show that, unlike the intrinsic parameters of the denoiser, the new scaling parameter can be explicitly related to the trade-off between the data-fidelity and the prior.

\item We extensively validate denoiser scaling by showing its potential to address the suboptimal performance of denoising CNNs within PnP algorithms. Our results show that denoiser scaling is a simple yet effective approach for boosting the performance of CNN priors within PnP.

\end{itemize}

\begin{figure*}
\begin{minipage}[t]{.5\textwidth}
\begin{algorithm}[H]
\caption{$\mathsf{PnP}$-$\mathsf{ADMM}$}\label{alg:pnpadmm}
\begin{algorithmic}[1]
\State \textbf{input: } $\xbm^0$, $\sbm^0 = \zerobm$, $\gamma > 0$, and $\mu > 0$
\For{$k = 1, 2, \dots$}
\State $\zbm^k = \prox_{\gamma g}(\xbm^{k-1} + \sbm^{k-1})$
\State $\xbm^k = \Dsf_\sigma(\zbm^k - \sbm^{k-1})$
\State $\sbm^k = \sbm^{k-1} + (\xbm^k - \zbm^k)$
\EndFor\label{euclidendwhile}
\end{algorithmic}
\end{algorithm}%
\end{minipage}%
\hspace{0.25em}
\begin{minipage}[t]{.5\textwidth}
\begin{algorithm}[H]
\caption{$\mathsf{PnP}$-$\mathsf{ISTA}/\mathsf{PnP}$-$\mathsf{FISTA}$}\label{alg:pnpista}
\begin{algorithmic}[1]
\State \textbf{input: } $\xbm^0 = \sbm^0$, $\gamma > 0$, $\mu > 0$, and $\{q_k\}_{k \in \N}$
\For{$k = 1, 2, \dots$}
\State $\zbm^k = \sbm^{k-1}-\gamma \nabla g(\sbm^{k-1})$
\State $\xbm^k = \Dsf_\sigma(\zbm^k)$
\State $\sbm^k = \xbm^k + ((q_{k-1}-1)/q_k)(\xbm^k-\xbm^{k-1}) $
\EndFor\label{euclidendwhile}
\end{algorithmic}
\end{algorithm}%
\end{minipage}
\end{figure*}


\section{Background}
\label{Sec:Background}

Consider the recovery of an unknown image $\xbm \in \R^n$ from noisy measurements $\ybm \in \R^m$. When this inverse problem is ill-posed, it is essential to include prior information on the unknown image to regularize the solution. A common approach is to formulate the problem as regularized inversion, expressed as an optimization problem of the form
\begin{equation}
\label{Eq:Optimization}
\xbmhat = \argmin_{\xbm \in \R^n} f(\xbm) \quad\text{with}\quad f(\xbm) = g(\xbm)
 + \lambda h(\xbm),\end{equation}
 where $g$ is the data-fidelity term, $h$ is the regularizer, and $\lambda > 0$ is a regularization parameter that adjusts their relative strengths. For example, by setting
 $$g(\xbm) = -\log(p_{\ybm|\xbm}(\ybm|\xbm))\quad\text{and}\quad  h(\xbm) = -(1/\lambda)\log(p_\xbm(\xbm)),$$
where $p_{\ybm | \xbm}$ denotes the likelihood function characterizing the imaging system and $p_\xbm$ denotes a probability distribution over $\xbm$, one obtains the classical \emph{maximum-a-posteriori probability (MAP)} estimator. 
The regularized inversion framework~\eqref{Eq:Optimization} can accommodate a variety of data-fidelity and regularization terms. For example, the combination of a linear measurement model under an additive white Gaussian noise (AWGN) and isotropic \emph{total variation (TV)} regularization~\cite{Rudin.etal1992, Beck.Teboulle2009a} is obtained by setting
$$g(\xbm) = \frac{1}{2}\|\ybm-\Hbm\xbm\|_2^2 \quad\text{and}\quad h(\xbm) = \sum_{i = 1}^n \|[\Dbm\xbm]_i\|_2,$$
where $\Dbm$ is the discrete image gradient. Isotropic TV corresponds to the sparsity-promoting $\ell_1$-norm prior on the magnitude of the image gradient.

A large number of regularizers used in the context of imaging inverse problems, including $\ell_1$-norm and TV, are nonsmooth. Proximal algorithms~\cite{Parikh.Boyd2014} enable efficient minimization of nonsmooth functions, without differentiating them, by using the \emph{proximal operator}, defined as
\begin{equation}
\label{Eq:ProximalOperator}
\prox_{\tau h}(\zbm) \defn \argmin_{\xbm \in \R^n}\left\{\frac{1}{2}\|\xbm-\zbm\|_2^2 + \tau h(\xbm)\right\},
\end{equation}
where $\tau > 0$ is a scaling parameter that controls the influence of $h$. Note that the proximal operator can be interpreted as a MAP image denoiser for AWGN with variance of $\tau$.

The observation that the proximal operator is an image denoiser for AWGN  prompted the development of PnP~\cite{Venkatakrishnan.etal2013}, where the operator $\prox_{\tau h}(\cdot)$, within a proximal algorithm, is replaced with a more general image denoiser $\Dsf(\cdot)$, such as BM3D~\cite{Dabov.etal2007} or DnCNN~\cite{Zhang.etal2017}.  The PnP-ADMM algorithm~\cite{Venkatakrishnan.etal2013} is summarized in  Algorithm~\ref{alg:pnpadmm}, where in analogy to $\prox_{\tau h}(\cdot)$ we also introduce the parameter $\sigma > 0$ for the denoiser $\Dsf_\sigma(\cdot)$.
An alternative algorithm, PnP-ISTA~\cite{Kamilov.etal2017}, is summarized Algorithm~\ref{alg:pnpista}. When the values for $\{q_k\}$ in this algorithm are set such that $q_k = 1$ for all $k \geq 1$, the algorithm corresponds to the traditional form of ISTA~\cite{Beck.Teboulle2009a}. Alternatively, when the values for $\{q_k\}$ are adapted as
\begin{equation}
\label{Eq:Accelerate}
q_k = \frac{1}{2}\left(1+\sqrt{1+4q_{k-1}^2}\right),
\end{equation}
the algorithm corresponds to the accelerated variant of ISTA, known as Fast ISTA (FISTA)~\cite{Beck.Teboulle2009}.
In this paper, we will use the sequence $\{q_k\}$ as a mechanism for switching between the methods. As extensively discussed in~\cite{Sun.etal2018a}, both PnP-ADMM and PnP-ISTA have the same set of fixed points, but the algorithms offer complimentary strategies for handling the data-fidelity term, and have different convergence behaviour. In particular, PnP-ADMM is known to be fast for forward operators that can be inverted efficiently~\cite{Matakos.etal2013, Almeida.Figueiredo2013, Wohlberg2016}, while PnP-ISTA is well suited to nonlinear forward models where $\prox_{\gamma g}(\cdot)$ is computationally expensive to evaluate~\cite{Kamilov.etal2017, Kamilov.etal2016}.


In traditional proximal optimization, the scaling parameter $\tau$ of the proximal operator~\eqref{Eq:ProximalOperator} is directly related to the regularization parameter $\lambda$. For example, by setting $\tau = \gamma \lambda$ within traditional ADMM or FISTA, one minimizes the objective function in~\eqref{Eq:Optimization}. However, this explicit relationship between the scaling parameter and the regularization parameter is lost in the context of more general denoisers. Since some popular image denoisers, such as BM3D, accept a parameter corresponding to the noise variance, current PnP algorithms generally treat it as a proxy for the regularization parameter.
For example, if $\sigma$ in the notation for $\Dsf_\sigma(\cdot)$ in Algorithm~\ref{alg:pnpadmm} and Algorithm~\ref{alg:pnpista} denotes the standard deviation parameter accepted by the denoiser, the common strategy is to set it as $\sigma = \sqrt{\gamma \lambda}$~\cite{Chan.etal2016}. However, this strategy does not work with all denoisers, since some do not have a dedicated parameter for noise variance. In particular, many denoising CNNs do not have a parameter for the noise standard deviation, which is often addressed by training multiple neural nets at different noise levels and using $\sigma$ to select the most suitable one for a given problem. The denoiser scaling technique introduced in the next section enables the control of the regularization strength for denoisers that have no intrinsic parameters analogous to $\sigma$.


\section{Proposed Method}

We introduce denoiser scaling for explicitly controlling the regularization strength in PnP. Remarkably, the technique can be theoretically justified from multiple perspectives, including from that of proximal optimization, statistical estimation, and consensus equilibrium~\cite{Buzzard.etal2017}. Our experimental results in Section~\ref{Sec:Experiments} corroborate the ability of denoiser scaling to control the relative influence of the denoiser.

\subsection{Denoiser scaling}

Consider an image denoiser $\Dsf: \R^n \rightarrow \R^n$, where we omit the parameter $\sigma$ from the notation as it is not available for all denoisers. We define the scaled denoiser as
\begin{equation}
\label{Eq:ScalingTechnique}
\Dsf_\mu(\zbm) \defn (1/\mu) \Dsf(\mu \zbm),\quad \zbm \in \R^n,
\end{equation}
where we will refer to the parameter $\mu > 0$ as the \emph{denoiser scaling parameter}. In the rest of this paper, we assume that the denoisers in line 4 of Algorithm~\ref{alg:pnpadmm} and Algorithm~\ref{alg:pnpista} correspond to the scaled denoiser $\Dsf_\mu(\cdot)$. Note that the scaling in~\eqref{Eq:ScalingTechnique} is \emph{complimentary} to any intrinsic parameter of $\Dsf(\cdot)$. For example, if the underlying denoiser $\Dsf(\cdot)$ additionally accepts $\sigma$ as a parameter, $\Dsf_\mu(\cdot)$ will also accept the same parameter. However, as discussed below, the parameter $\mu$ will enable control of the strength of regularization when $\sigma$ is not available.

\subsection{Proximal operator denoisers}

We first consider the case in which $\Dsf(\cdot)$ is an \emph{implicit} proximal operator of some \emph{unknown} $h$, which is a common interpretation for PnP algorithms~\cite{Sreehari.etal2016}. For convenience, we assume $h$ to be closed, convex, and proper~\cite{Parikh.Boyd2014}; however, this assumption can be dropped as long as $\prox_h(\cdot)$ is well defined for the given $h$. We state the following result for the scaled denoisers.
\begin{proposition}
\label{Prop:ProxOp}
Suppose ${\Dsf(\zbm) = \prox_h(\zbm)}$, where $h$ is a closed, convex, and proper function. Then, we have
\begin{equation}
\Dsf_\mu(\zbm) = \prox_{\mu^{-2}h(\mu \cdot)}(\zbm), \quad \zbm \in \R^n.
\end{equation}
\end{proposition}

\begin{proof}
See Appendix~\ref{Sec:ProofProxDen}.
\end{proof}

Proposition~\ref{Prop:ProxOp} indicates that by scaling an implicit proximal operator, one directly adjusts the strength of regularization via the scaling of the regularizer by $1/\mu^2$ and of the input to $h$ by $\mu$. While the relationship between $\mu$ and the regularization parameter in front of $h$ is not linear, $\mu$ still provides an explicit mechanism to tune the denoiser. If the denoiser corresponds to a \emph{$1$-homogeneous} regularizer $h$, we have $h(\mu \cdot) = \mu \cdot h(\cdot)$, which directly implies
\begin{equation}
\label{Eq:LinScal}
\Dsf_\mu(\zbm) = \prox_{\mu^{-1} h}(\zbm), \quad \zbm \in \R^n.
\end{equation}
This means that the denoiser scaling becomes equivalent to tuning the traditional regularization parameter in regularized inversion. Since any norm and semi-norm is $1$-homogeneous, the strength of many implicit and explicit regularizers, such as the $\ell_1$-norm or TV penalty, can be directly adjusted through denoiser scaling.
This equivalence is confirmed numerically for the TV denoiser in Section~\ref{Sec:Experiments}.

\subsection{Mean squared error optimal denoisers}

We now consider the case of a denoiser that performs the minimum mean-squared error (MMSE) estimation of a vector from its AWGN corrupted version~\cite{Kamilov.etal2013, Kazerouni.etal2013, Gribonval.Machart2013}. MMSE denoisers are optimal with respect to the ubiquitous image-quality metrics, such as signal-to-noise ratio (SNR). Additionally, many popular denoisers (such as BM3D and certain denoising CNNs) are often interpreted as \emph{empirical} MMSE denoisers. We state the following result for the scaled MMSE denoisers.

\begin{proposition}
\label{Prop:MMSEden}
Suppose $\Dsf(\cdot)$ computes the MMSE solution of the following denoising problem
$$\zbm = \xbm + \nbm \quad\text{with}\quad \xbm \sim p_\xbm\quad\text{and}\quad \nbm \sim \Ncal(\zerobm, \Ibm),$$
where ${\Ibm}$ is an identity matrix. Then, the denoiser~\eqref{Eq:ScalingTechnique} computes the MMSE solution of
$$
\zbm = \ubm + \ebm \quad\text{with}\quad \ubm \sim p_\xbm(\mu\cdot)\quad\text{and}\quad \ebm \sim \Ncal(\zerobm, \mu^{-2}\Ibm).
$$
\end{proposition}

\begin{proof}
See Appendix~\ref{Sec:ProofMMSEDen}.
\end{proof}

Similarly to Proposition~\ref{Prop:ProxOp},
Proposition~\ref{Prop:MMSEden} indicates that the scaling of the MMSE denoiser enables the direct control of the strength of regularization via the scaling of the noise variance by $1/\mu^2$ and of the input to the prior $p_\xbm$ by $\mu$. If $p_\xbm$ assigns equal probabilities to all images that have undergone rescaling, we have $p_\xbm(\mu \cdot) = p_\xbm(\cdot)$, which implies that the scaled denoiser directly adjusts the variance of AWGN in the MMSE estimation.


\begin{figure}[t]
        \centering\includegraphics[width=12cm]{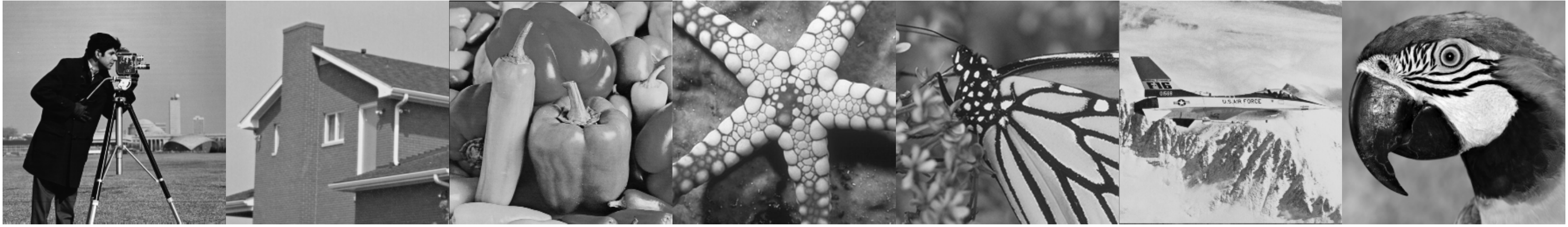}
        \caption{Test images used for the quantitative performance evaluation. From left to right: \textit{Cameraman}, \textit{House},  \textit{Pepper},  \textit{Starfish},  \textit{Butterfly},  \textit{Plane},  \textit{Parrot}.}
        \label{Fig:testimg}
\end{figure}

\begin{table*}[t]
        \centering
        \caption{SNR performances of several image denoisers at different noise levels. }
        \vspace{5pt}
        \footnotesize
        \label{Tab:denoise}
        \begin{tabular*}{16.5cm}{M{40pt}	cC{35pt}C{35pt}		cC{35pt}C{35pt}C{35pt}		cC{35pt}C{35pt}C{35pt}	}
                \toprule
                \multirow{2}{*}{\textbf{Input SNR}}	&&\multicolumn{2}{c}{\textbf{TV}}	&&\multicolumn{3}{c}{\textbf{BM3D}}		&&\multicolumn{3}{c}{\textbf{$\DnCNNast$}}\\
                \cmidrule{3-4}\cmidrule{6-8}\cmidrule{10-12}

                &&\textbf{Scaled}&\textbf{Optimized}
                &&\textbf{Unscaled}&\textbf{Scaled} &\textbf{Optimized}
                &&\textbf{Unscaled}&\textbf{Scaled} &\textbf{Optimized}\\

                \cmidrule{1-12}
                \textbf{15 dB}  &&22.77 &22.77		&&16.58 &24.55&24.51		&&16.49 &24.30&24.21\\
                \textbf{20 dB}  &&25.80 &25.80		&&25.58 &27.34&27.34		&&24.83 &27.42&27.23\\
                \textbf{25 dB}  &&29.09 &29.09		&&29.61 &30.43&30.43		&&29.83 &30.51&30.35 \\
                \textbf{30 dB}  &&32.66 &32.66		&&33.49 &33.63&33.73		&&33.54 &33.82&33.72\\

                \bottomrule

        \end{tabular*}
\end{table*}

\begin{figure*}[t]
        \centering\includegraphics[width=\textwidth]{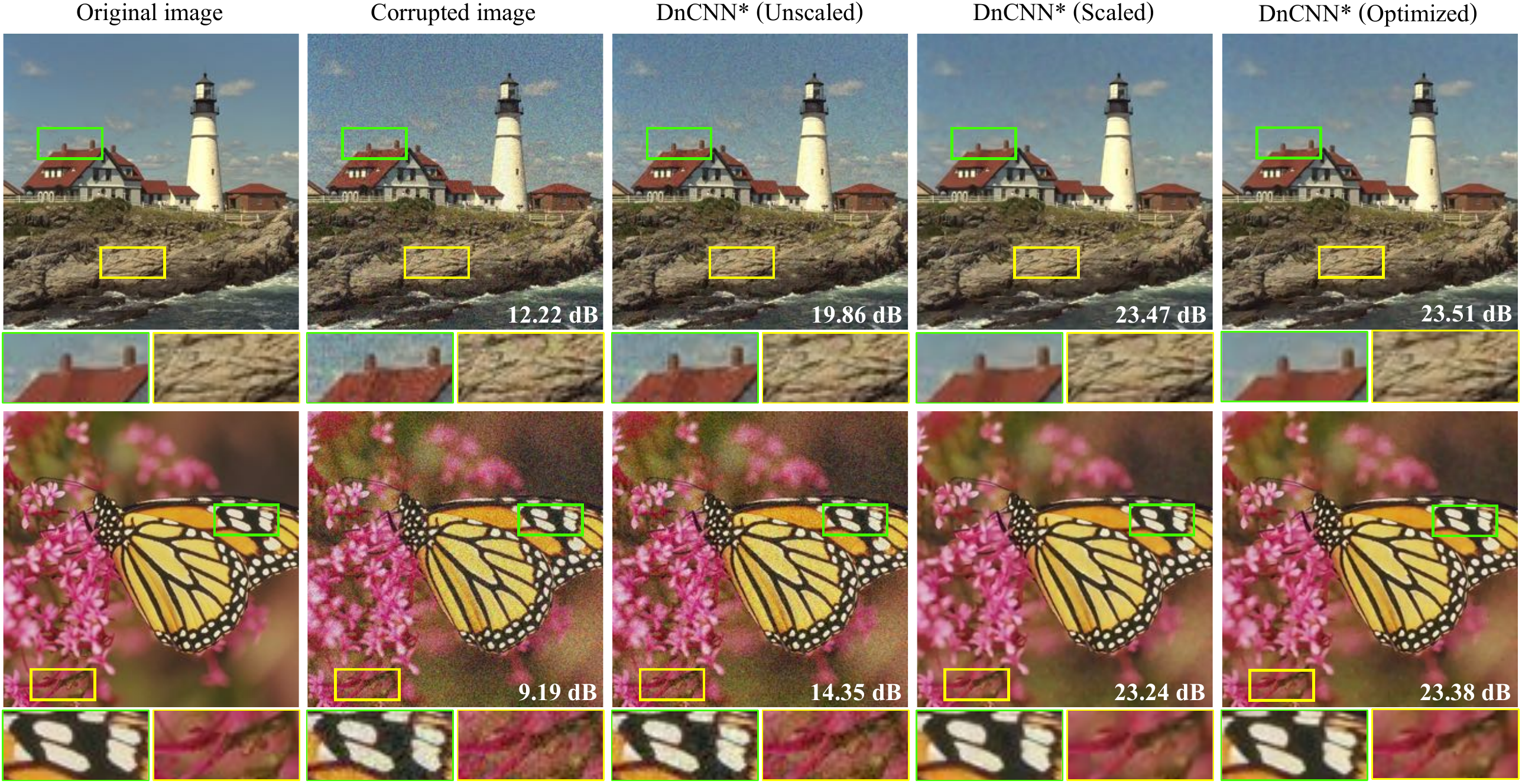}
        \caption{Illustration of denoiser scaling on two color images: \textit{Lighthouse} (top) and \textit{Monarch} (bottom). The noise levels in the top and bottom rows are $\sigma = 30$ and $\sigma = 40$, respectively. $\DnCNNast$ (Optimized) corresponds to the CNN denoiser trained using the correct noise levels. On the other hand, $\DnCNNast$ (Unscaled) and $\DnCNNast$ (Scaled) use the same CNN trained at a mismatched noise level of $\sigma = 20$. By adjusting $\mu$, $\DnCNNast$ trained at a suboptimal $\sigma$ can be made to match the performance of $\DnCNNast$ trained using the correct noise level.}
        \label{Fig:denoise}
\end{figure*}

\subsection{Consensus equilibrium interpretation}

Consensus equilibrium (CE)~\cite{Buzzard.etal2017} is a recent framework for interpreting the solution of regularized inversion methods in terms of a set of balancing equations for the
forward and prior models, without an explicit cost function The solutions obtained by both PnP-ADMM and PnP-ISTA can be expressed in terms of the same set of CE equations
\begin{subequations}
\begin{align}
&\xbm = \Fsf(\xbm + \sbm) \\
&\xbm = \Dsf(\xbm - \sbm),
\end{align}
\end{subequations}
where $\Fsf(\cdot) \defn \prox_{\gamma g}(\cdot)$ and  $\gamma > 0$ is an algorithm tuning parameter. We use the CE framework to state the following result for the scaled denoisers.

\begin{proposition}
\label{Prop:CEden}
Let $g$ be a smooth, convex function and $\Dsf(\cdot)$ be a continuous denoiser. The fixed point $(\xbm, \sbm)$ of PnP-ADMM and PnP-ISTA for the scaled denoiser satisfies
$$
\begin{dcases}
\mu\xbm = \prox_{(\gamma\mu^2)g(\cdot/\mu)}\left(\mu\xbm + \sbm\right)\\
\mu\xbm = \Dsf\left(\mu\xbm - \sbm\right).
\end{dcases}
$$
\end{proposition}

\begin{proof}
        See Appendix~\ref{Sec:ProofCEden}.
\end{proof}
This result establishes a direct relationship between the scaling parameter $\mu > 0$ and the rescaling of $g$ with respect to the denoiser. Note that while Proposition~\ref{Prop:ProxOp} and~\ref{Prop:MMSEden} discuss the impact of denoisier scaling on the implicit prior, Proposition~\ref{Prop:CEden} highlights its impact on the relative influence between the denoiser and the data-fidelity term via the weighting in front of $g$. Since our only assumption is the continuity of the denoiser, Proposition~\ref{Prop:CEden} also relaxes the assumptions on the denoiser. While the relationship between $\mu$ and the set of equilibrium points is nontrivial, the proposition implies that one can still adjust the amount of regularization by tuning $\mu$.
\begin{figure*}[t]
        \centering\includegraphics[width=\textwidth]{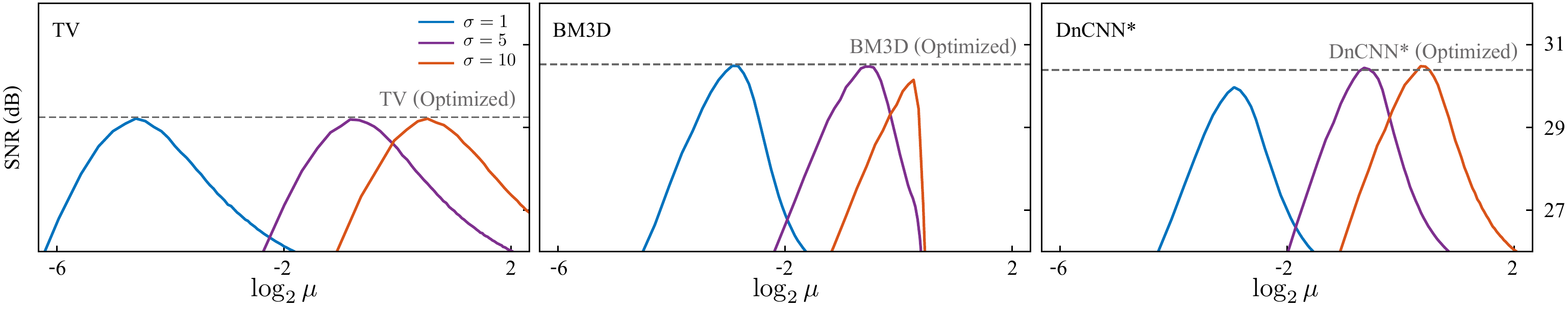}
        \caption{The influence of the scaling parameter $\mu$ on the denoising performance for \textit{Pepper} for AWGN with input SNR of 25 dB ($\sigma = 7.23$). We show the SNR evolution against $\mu$ for the variants of TV, BM3D, and $\DnCNNast$ designed for the mismatched noise levels. The horizontal line shows the performance of the corresponding denoiser optimized for input SNR of 25 dB. Note how by adjusting $\mu$, one can achieve nearly optimal performance for all three denoisers.}
        \label{Fig:scale}
\end{figure*}

\begin{figure*}[t]
        \centering\includegraphics[width=\textwidth]{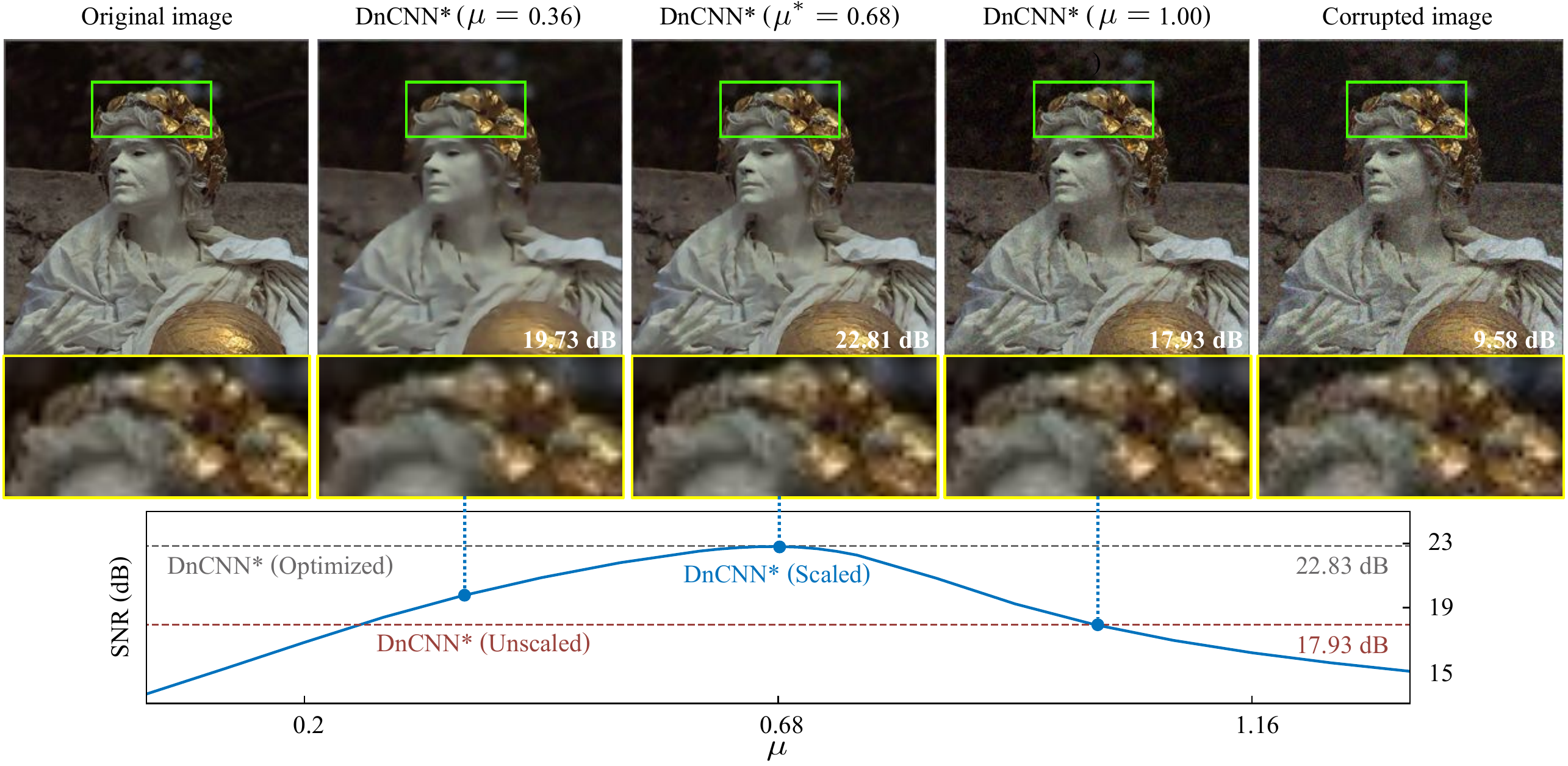}
        \caption{The influence of the denoiser scaling parameter $\mu$ on the denoising performance of $\DnCNNast$ on the color image \textit{Statue}. The noise in the image corresponds to $\sigma = 30$, while $\DnCNNast$ was trained for the removal of noise corresponding to $\sigma = 20$. The top row images illustrate the visual performance at $\mu$ values of $0.36$, $0.68$, and $1.00$. The bottom plot shows the SNR evolution against the parameter $\mu$ for a wider range of values. The scaled $\DnCNNast$ achieves its best performance at $\mu^\ast = 0.68$. Note how unscaled $\DnCNNast$ ($\mu = 1.00$) leads to an insufficient amount of regularization, while a smaller scaling parameter $\mu = 0.36$ leads to oversmoothing. This figure highlights the ability of $\mu$ to control the strength of regularization with a CNN denoiser.}
        \label{Fig:evolution}
\end{figure*}

\section{Numerical Validation}
\label{Sec:Experiments}







\begin{table*}[t]
        \centering
        \caption{Average SNRs obtained for different inverse problems and image denoisers.}
        \vspace{5pt}
        \footnotesize
        \label{Tab:mri_sr}
        \begin{tabular*}{16.5cm}{M{65pt}M{40pt}	cC{30pt}C{30pt}		cC{35pt}C{30pt}C{30pt}		cC{35pt}C{30pt}}
                \toprule

                        \multirow{2}{*}{\textbf{Forward model}}& \multirow{2}{*}{\textbf{Input SNR} }
                         &&\multicolumn{2}{c}{\textbf{TV}}
                         &&\multicolumn{3}{c}{\textbf{BM3D}}
                         &&\multicolumn{2}{c}{\textbf{$\DnCNNast$}}\\
                        \cmidrule{4-5}\cmidrule{7-9}\cmidrule{11-12}

                        &
                        &&\textbf{Scaled}&\textbf{Optim}
                        &&\textbf{Unscaled}&\textbf{Scaled} & \textbf{Optim}
                        &&\textbf{Unscaled}&\textbf{Scaled}      \\

                        \cmidrule{1-12}

                \multirow{2}{*}{\textbf{Fourier}}
                & \textbf{30 dB}        &&27.15 &27.15	&&27.29 &28.49&28.51	&&27.92 &28.44 \\
                & \textbf{40 dB}        &&27.79 &27.79	&&28.97 &29.17&29.21	&&29.72 &29.89 \\
                \noalign{\vskip 1pt}\cdashline{1-12}\noalign{\vskip 3pt}

                \multirow{2}{*}{\textbf{Super-resolution}}
                & \textbf{30 dB}        &&19.86 &19.87	&&14.32&20.65&20.64		&&13.59 &20.05 \\
                & \textbf{40 dB}        &&22.42 &22.41	&&22.64&22.78&22.71		&&23.03 &23.07 \\
                \bottomrule
        \end{tabular*}
\end{table*}

\begin{figure*}[t]
        \centering\includegraphics[width=\textwidth]{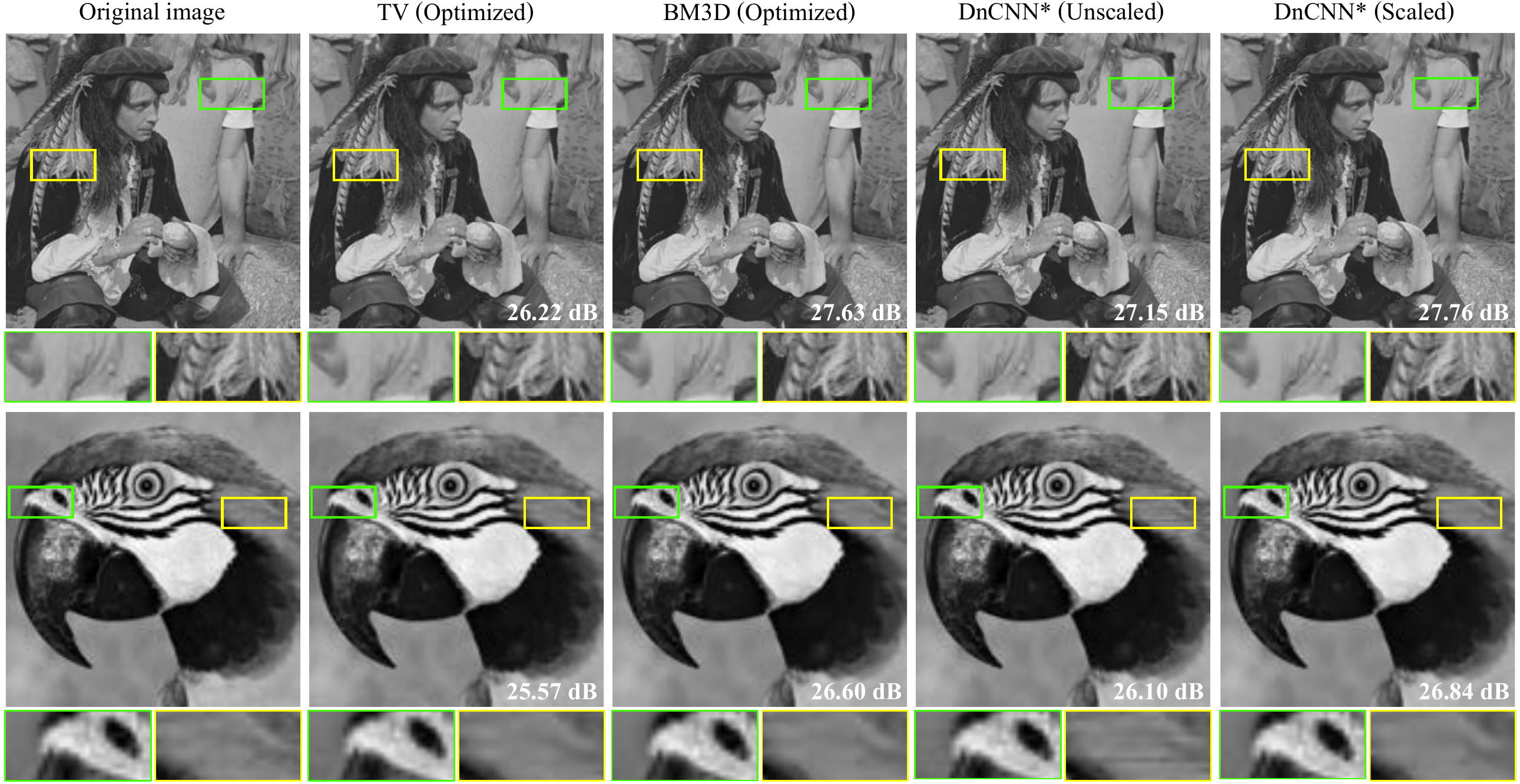}
        \caption{Visual illustration of denoiser scaling on the subsampled Fourier operator and two natural images: \textit{Man} (top) and \textit{Parrot} (bottom). The sampling rate is $m/n = 1/3$ and input SNR is 30 dB. $\DnCNNast$ is selected from $\sigma \in \{1, 5, 10\}$ that produces the best SNR performance. $\DnCNNast$ (Scaled) relies on the same CNN selected by $\DnCNNast$ (Unscaled). This result illustrates the potential of $\DnCNNast$ (Scaled) to boost the performance of $\DnCNNast$ (Unscaled) by reducing the artifacts.}
        \label{Fig:mri_ni}
\end{figure*}

\begin{figure*}[t]
	\centering\includegraphics[width=\textwidth]{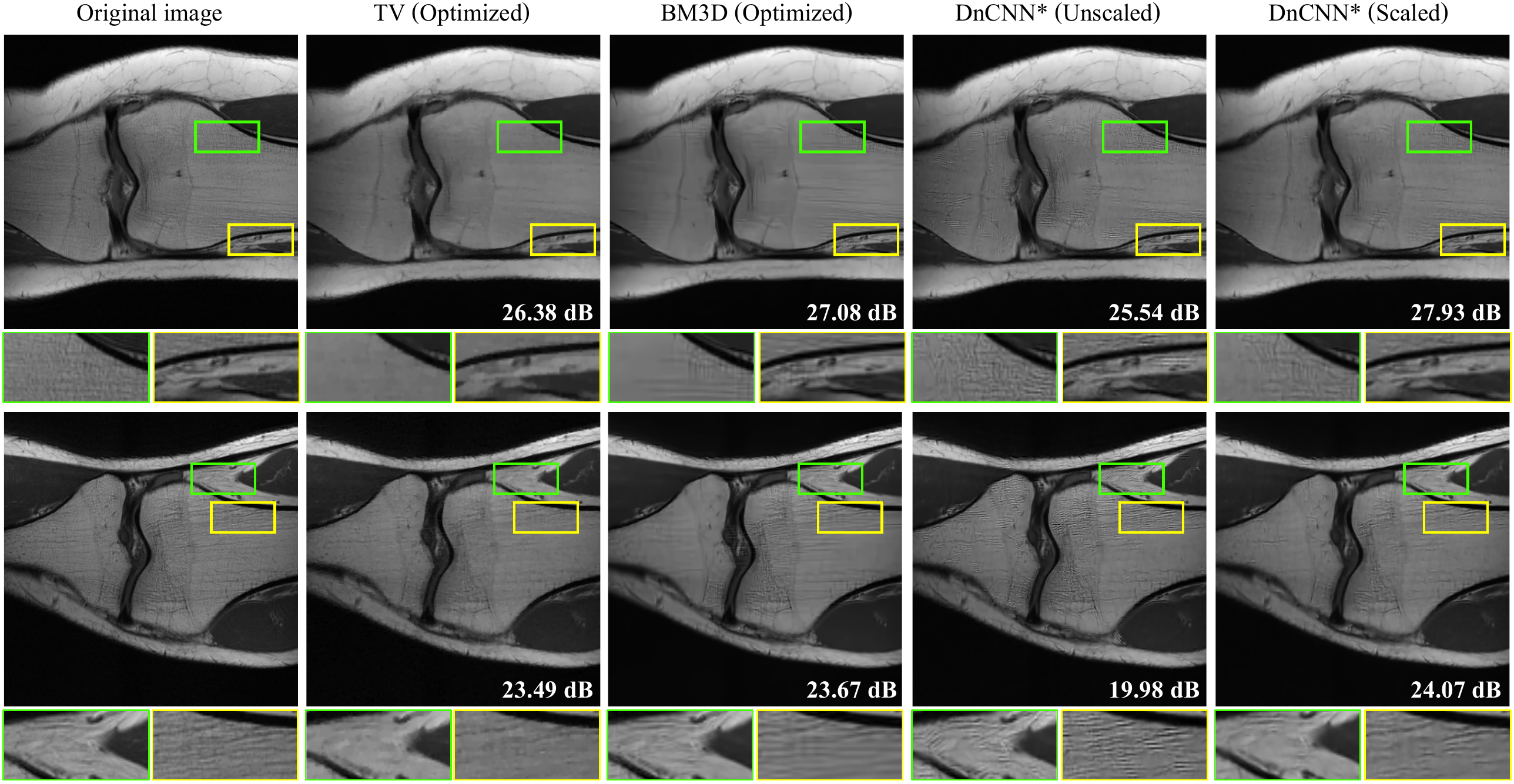}
	\caption{Visual illustration of denoiser scaling on the subsampled Fourier operator and two medical images: \textit{Knee 1} (top) and \textit{Knee 2} (bottom) from the fastMRI dataset. The sampling ration is $m/n = 1/3$ and input SNR is 30 dB. $\DnCNNast$ is selected from $\sigma \in \{1, 5, 10\}$ that produces the best SNR performance. $\DnCNNast$ (Scaled) relies on the same CNN selected by $\DnCNNast$ (Unscaled). Note how $\DnCNNast$ (Scaled) improves the visual quality of results compared to $\DnCNNast$ (Unscaled).}
	\label{Fig:mri_ki}
\end{figure*}
\begin{figure*}[t]
        \centering\includegraphics[width=\textwidth]{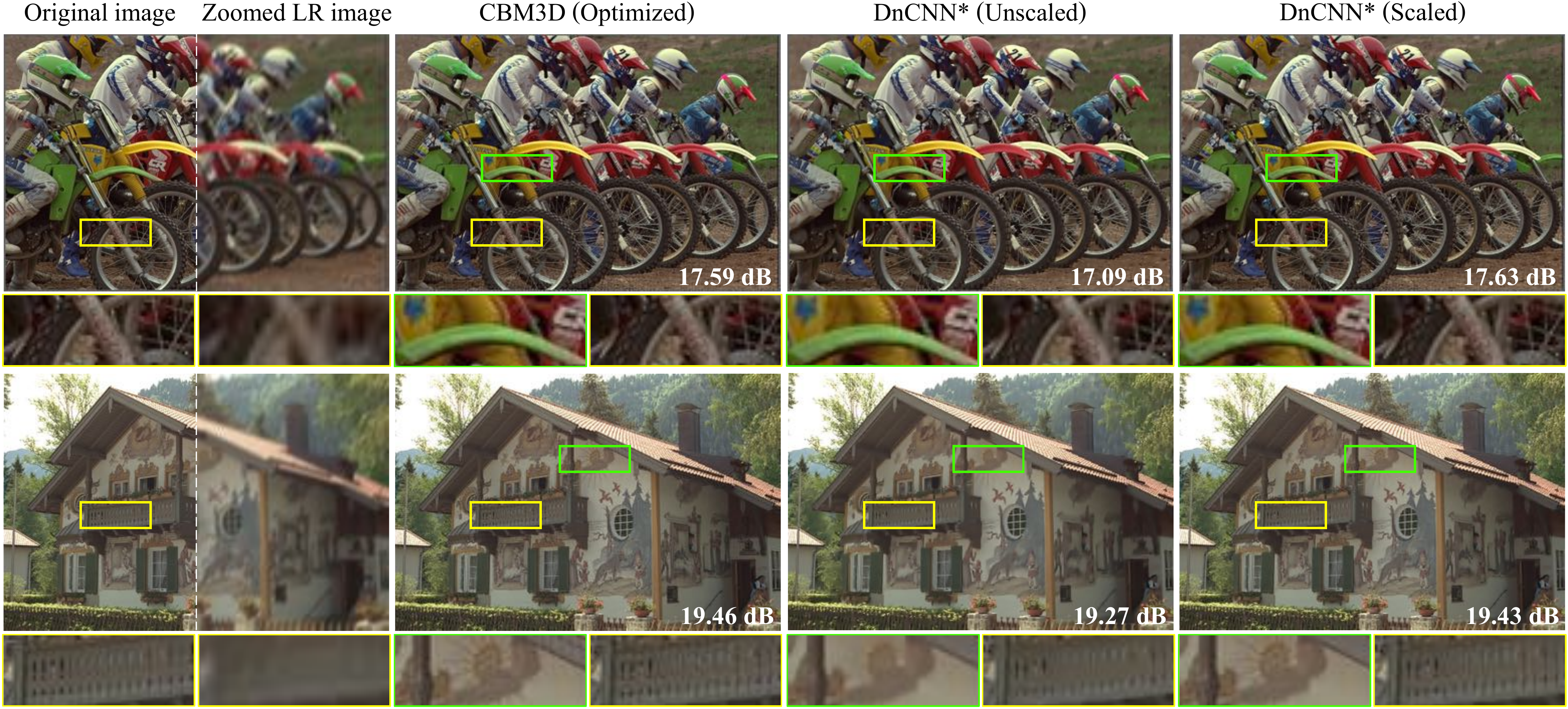}
        \caption{Results of SR simulation on color image \textit{Bikes} (top) and \textit{Paintedhouse} (bottom). The input SNR is 40 dB. $\DnCNNast$ is selected from $\sigma \in \{1, 5, 10\}$ that produces the best SNR performance. $\DnCNNast$ (Scaled) relies on the same CNN selected by $\DnCNNast$ (Unscaled). One can see while $\DnCNNast$ blur out the details in some regions,  $\DnCNNast$ (Scaled) can generate images with more details and sharper edges.}
        \label{Fig:sr}
\end{figure*}

In this section, we demonstrate the ability of denoiser scaling to boost the performance of PnP. Our experiments consider imaging inverse problems of the form ${\ybm = \Hbm \xbm + \ebm}$, where  $\ybm \in  \R^m$ denotes the measurements, $\ebm \in \R^m$ is a vector of AWGN with zero mean and standard deviation $\sigma$, and ${\Hbm \in \R^{m \times n}}$ denotes the forward operator.  We focus on three inverse problems: image denoising, subsampled Fourier inversion, which commonly used in magnetic resonance imaging (MRI), and single image super-resolution (SR), where the forward models correspond to the identity matrix, radially subsampled two-dimensional Fourier transform, and blurring-downsampling operator, respectively. For each simulation, the measurements are corrupted with AWGN quantified through the input signal-to-noise ratio (SNR). We will consider three denoisers for PnP: TV, BM3D, and our own residual $\DnCNNast$. $\DnCNNast$ is a simplified variant of the standard DnCNN~\cite{Zhang.etal2017}, where our simplifications correspond to the removal of batch-normalization layers and reduction in the total number of layers (see Appendix~\ref{Sec:CNNstructure} for details). This simplification reduces the computational cost of applying the denoiser across multiple PnP iterations. We use $\sigma$ to parameterize BM3D and $\DnCNNast$. For BM3D $\sigma$ represents the parameter of the denoiser represenging the standard deviation of noise and for $\DnCNNast$ it represents the standard deviation of the noise used for training the CNN. We follow~\cite{Chen.Pock2016} and use 400 images of size $180 \times 180$ to train three $\DnCNNast$ instances, on natural grayscale images, for the removal of AWGN at three noise levels, $\sigma \in \{1, 5, 10\}$. For some of the experiments, we also use a constrained variant of BM3D, where $\sigma$ is restricted to the same set. For the medical knee images,  we follow the work in~\cite{Sun.etal2019b} to train our 7-layer $\DnCNNast$ on NYU fastMRI dataset~\cite{Zbontar.etal2018} for $\sigma \in \{1, 5, 10\}$. For color images, we use $\DnCNNast$ trained on 4744 images from the Waterloo Exploration Database~\cite{ma2016waterloo} and use the CBM3D denoiser~\cite{Dabov.etal2007a}.  For all experiments in SR and Fourier, we use PnP-FISTA as the reconstruction algorithm. All the quantitative results in the tables are averaged over seven test images shown in Figure~\ref{Fig:testimg} with hyperparameters optimized individually for each image for the best SNR performance using grid search.  All visual results are shown with SNR value displayed directly on the images. None of the test images were used in training.

Table ~\ref{Tab:denoise} shows the ability of denoiser scaling to adjust the denoising strength for three different image denoisers at four input SNR levels: 15 dB, 20 dB, 25 dB and 30 dB. In the table, TV (Optimized) is obtained by tuning the regularization parameter $\lambda$ for each test image. For TV (Scaled), we fix $\lambda = 1$ and tune the scaling parameter $\mu$ for the best result. The performances of BM3D (Unscaled and Scaled) and $\DnCNNast$ (Unscaled and Scaled) correspond to the best instance selected from the limited set of $\sigma \in \{1, 5, 10\}$. For reference, we also show BM3D (Optimized), which uses fully optimized $\sigma$, and DnCNN (Optimized), which is trained on noisy images with the true input SNR. Table~\ref{Tab:denoise} highlights the equivalence between TV with an optimized $\lambda$ and TV with $\lambda = 1$, but optimized $\mu$, which validates eq.~\eqref{Eq:LinScal}. Table~\ref{Tab:denoise} also shows that denoiser scaling significantly improves the performance of sub-optimally tuned BM3D and $\DnCNNast$ to achieve the performance of the corresponding denoiser with optimized $\sigma$.

Figure~\ref{Fig:denoise} visually illustrates the performance of denoiser scaling on the problem of color image denoising for AWGN of $\sigma = 30$ (top) and $\sigma = 40$ (bottom). In the figure, $\DnCNNast$ (Optimized) denotes to the denoiser trained using the dataset with the correct noise level. On the other hand, $\DnCNNast$ (Unscaled) and $\DnCNNast$ (Scaled) correspond to the same CNN instance, trained for noise level $\sigma = 20$. $\DnCNNast$ (Unscaled) uses $\mu = 1$, while $\DnCNNast$ (Scaled) optimizes $\mu$ for the best SNR performance. By simply adjusting $\mu$, the suboptimaly trained $\DnCNNast$ achieves the performance of $\DnCNNast$ trained using the correct noise level.


Figure~\ref{Fig:scale} considers the problem of denoising an image with the input SNR of 25 dB, which corresponds to the noise level $\sigma = 7.23$. The figure shows the influence of the parameter $\mu$ for improving the performance of denoisers at mismatched values of $\sigma = 1$, $5$, and $10$. The SNR value after denoising is plotted against the logarithm of $\mu$. For TV, we fixed the regularization parameter $\lambda$ to its optimal value for the noise levels $\sigma \in \{1, 5, 10\}$ and then adjusted the parameter $\mu$. For each plot, we also provide the performance of the denoisers with optimized $\sigma$. The comparison between unscaled denoisers at $\mu = 1$ and their scaled counterparts demonstrate the potential of denoiser scaling to influence the final performance, validating the theoretical conclusion that the denoiser scaling directly controls the strength of regularization. This is particularly appealing for $\DnCNNast$, which does not have a tunable parameter $\sigma$.

Figure~\ref{Fig:evolution} visually and quantitatively illustrates the influence of the parameter $\mu$ for color image denoising. The $\DnCNNast$ trained on noise level $\sigma = 20$ is applied for denoising an image with a noise level $\sigma = 30$. The performance of scaled DnCNN achieves its peak SNR at $\mu^\ast = 0.68$ and leads to either under- or over-regularization at either side of this value.

Table~\ref{Tab:mri_sr} shows results of image reconstruction for both Fourier and SR under two input noise levels 30 dB and 40 dB. Here, we fix $\sigma = 1$ for both BM3D and $\DnCNNast$. We aslo show BM3D (Optimized), which uses fully optimized $\sigma$. For Fourier, the measurement ratio is set to be approximately $m/n = 1/3$. For SR, the low resolution (LR) image is simulated by convolving the high resolution image (HR) with a motion-blur kernel of size $19 \times 19$ from~\cite{Liu.etal2019}, followed by down-sampling with scale factor 2. Note that TV (Scaled) has a fixed $\lambda= 1$ and a scaling parameter $\mu$ optimized for each test image. Some visual results are shown in Figures~\ref{Fig:mri_ni},~\ref{Fig:mri_ki}, and~\ref{Fig:sr}. These results clearly highlight the potential of denoiser scaling to significantly boost the performance of PnP. The results also highlight that the amount of improvement depends on the inverse problem in question (compare, for example, the results for SR under 30 dB and 40 dB of noise). This is expected since a specific regularizer might be already well suited for a given forward model (even without any denoiser scaling).


\section{Conclusion}


We have presented a simple, but effective, technique for improving the performance of PnP algorithms. The approach is justified theoretically by connecting the denoiser scaling with the strength of the effective regularization introduced by PnP. The proposed technique is shown to be particularly valuable when PnP is used with CNN denoisers that have no explicit tunable parameters. Our experimental results show the potential of denoiser scaling to significantly improve the performance of PnP algorithms across several inverse problems. While we have focused on PnP, the denoiser scaling approach can be applied more broadly to improve the performance of related methods, such as AMP and RED.


\appendix

\section{Proof of Proposition~\ref{Prop:ProxOp}}
\label{Sec:ProofProxDen}

The proof below is a direct consequence of definition~\eqref{Eq:ProximalOperator}. It is similar to the derivation of other properties of the proximal operator that have been extensively described in the literature (see for example~\cite[Ch. 6]{Beck2017a}). For any $\zbm \in \R^n$, we have that
\begin{align*}
&\Dsf_\mu(\zbm) = (1/\mu) \cdot \prox_{h}(\mu\cdot\zbm) \\
&=(1/\mu) \cdot \argmin_{\xbm} \left\{\frac{1}{2}\|\xbm-\mu\zbm\|_2^2 + h(\xbm)\right\} \\
&= (1/\mu) \cdot \argmin_{\xbm} \left\{\frac{1}{2}\left\|(\xbm/\mu)-\zbm\right\|_2^2 + (1/\mu^2) h(\xbm)\right\} \\
&= (1/\mu) \cdot \mu \cdot \argmin_{\ubm} \left\{\frac{1}{2}\|\ubm-\zbm\|_2^2 + (1/\mu^2)h(\mu\ubm )\right\} \\
& = \prox_{(1/\mu^2) h(\mu \cdot)}(\zbm),
\end{align*}
where in the second and the last lines we used the definition of the proximal operator, and in the fourth line we performed the variable change $\ubm = \xbm/\mu$.

\section{Proof of Proposition~\ref{Prop:MMSEden}}
\label{Sec:ProofMMSEDen}

This result is a direct consequence of the definition of the MMSE denoiser. For any $\zbm \in \R^n$, we have that
\begin{align*}
&\Dsf_\mu(\zbm) = (1/\mu) \cdot \Dsf(\mu\cdot\zbm) \\
&=(1/\mu) \cdot \frac{\int_{\R^n} \xbm \phi_1(\xbm-\mu\zbm) p_\xbm(\xbm) \d \xbm}{\int_{\R^n} \phi_1(\xbm-\mu\zbm) p_\xbm(\xbm) \d \xbm} \\
&=(1/\mu) \cdot \frac{\int_{\R^n} \xbm \phi_{(1/\mu^2)}(\xbm/\mu-\zbm) p_\xbm(\xbm) \d \xbm}{\int_{\R^n} \phi_{(1/\mu^2)}(\xbm/\mu-\zbm) p_\xbm(\xbm) \d \xbm}\\
&=(1/\mu) \cdot \frac{\int_{\R^n} (\mu\ubm) \phi_{(1/\mu^2)}(\ubm-\zbm) p_\xbm(\mu\ubm) \d \ubm}{\int_{\R^n} \phi_{(1/\mu^2)}(\ubm-\zbm) p_\xbm(\mu\ubm) \d \ubm},
\end{align*}
where we defined the probability density function of AWGN of variance $\nu > 0$ as
$$\phi_{\nu}(\xbm) \defn \frac{1}{\sqrt{2\pi\nu}}\e^{-\frac{\|\xbm\|^2}{2\nu}}.$$
The final line corresponds to the MMSE estimate of a random variable $\ubm \sim p_\xbm(\mu\cdot)$ from AWGN of variance $1/\mu^2$.

\section{Proof of Proposition~\ref{Prop:CEden}}
\label{Sec:ProofCEden}

It has already been shown that for any continuous denoiser both PnP-ADMM and PnP-ISTA have the same set of fixed points~\cite{Sun.etal2018a}. Consider any fixed point $\xbm$ of PnP-ISTA
\begin{align*}
&\xbm = \Dsf_\mu(\xbm-\gamma \nabla g(\xbm)) \\
&\Leftrightarrow\quad \mu\xbm = \Dsf(\mu\xbm -\gamma \mu \nabla g(\xbm)) \\
&\Leftrightarrow\quad
\begin{dcases}
\mu\xbm = \Dsf(\mu\xbm-\zbm)\\
\zbm = \gamma\mu\nabla g(\xbm)
\end{dcases} \\
&\Leftrightarrow\quad
\begin{dcases}
\mu\xbm = \Dsf(\mu\xbm-\zbm)\\
\mu\xbm = \prox_{(\gamma \mu^2)g(\cdot/\mu)}(\mu\xbm+\zbm),
\end{dcases}
\end{align*}
which directly leads to the result. To see the last equivalence, assume that $g$ is a smooth and convex function and note that
\begin{align*}
&\mu\xbm = \prox_{(\mu^2\gamma)g(\cdot/\mu)}\left(\mu\xbm + \zbm\right) \\
&= \argmin_{\ubm} \left\{\frac{1}{2}\left\|\ubm - \left(\mu\xbm + \zbm\right)\right\|_2^2 + (\gamma\mu^2)g(\ubm/\mu)\right\}\\
&\Leftrightarrow\quad \mu\xbm - (\mu\xbm + \zbm) + (\gamma\mu^2) \cdot (1/\mu) \cdot \nabla g(\mu\xbm/\mu) = \zerobm \\
&\Leftrightarrow\quad \zbm = \gamma\mu \nabla g(\xbm),
\end{align*}
where we used the optimality conditions.

\begin{figure}[t]
        \centering\includegraphics[width=8.5cm]{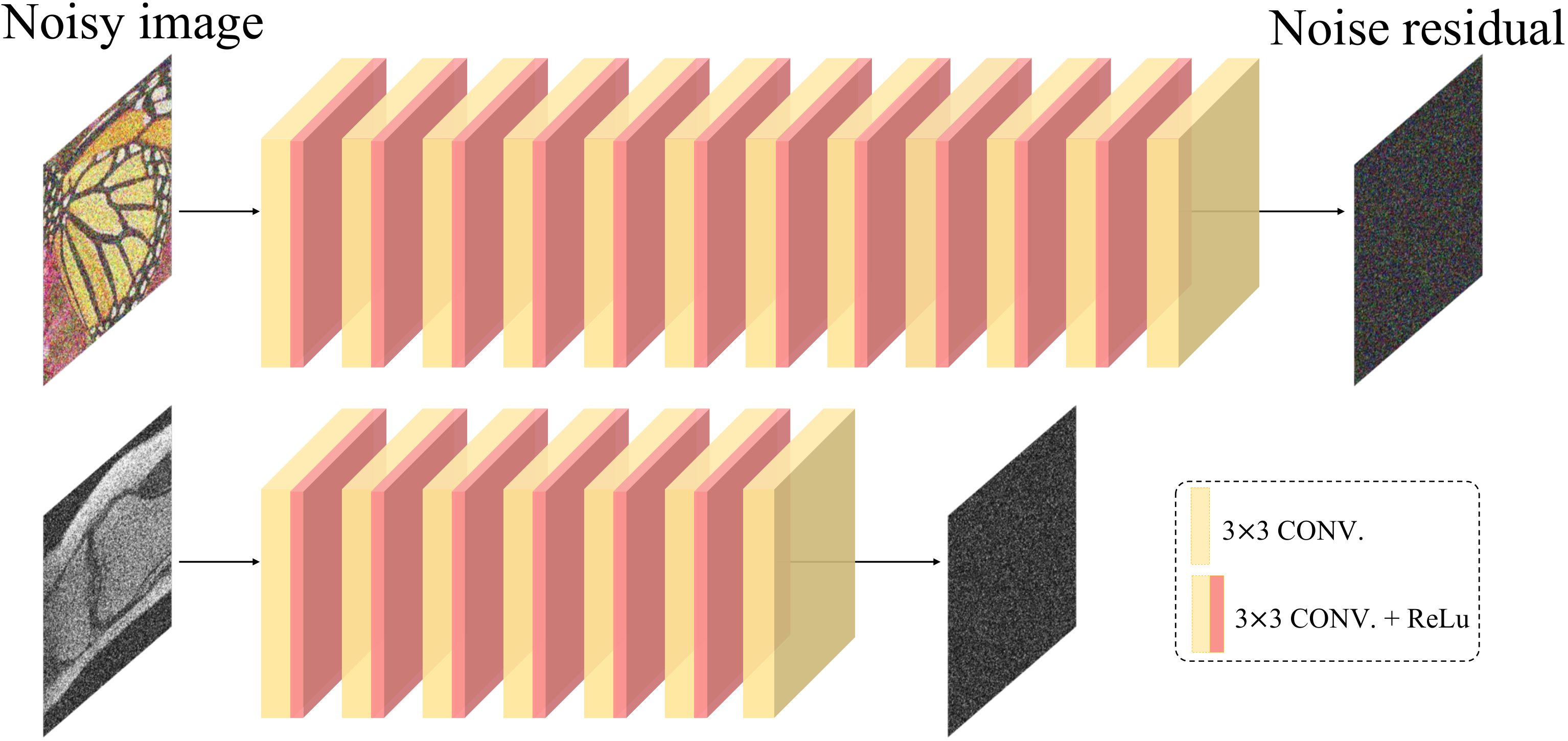}
        \caption{The architecture of two variants of $\DnCNNast$ we use in our simulations. $\DnCNNast$ (top) is applied on natural images, and $\DnCNNast$ (bottom) is applied on the medical knee images. Both neural nets are trained to predict the AWGN from the input. The final desired denoiser $\Dsf$ is obtained by simply subtracting the predicted noise from the input $\Dsf(\xbm) = \xbm - \DnCNNast(\xbm)$.}
        \label{Fig:net}
\end{figure}

\section{Architecture of the $\DnCNNast$ denoiser}
\label{Sec:CNNstructure}
Two variants of the residual $\DnCNNast$, shown in Figure~\ref{Fig:net}, are used in our simulations. The $\DnCNNast$ of 12 convolutional layers is used for natural images. The $\DnCNNast$ of 7 layers from~\cite{Sun.etal2019b} is used for the knee images from the NYU fastMRI dataset~\cite{Zbontar.etal2018}. The latter has a bounded Lipschitz constant $L = 2$, providing a necessary but not sufficient condition for $\Dsf$ to be a nonexpansive denoiser. As discussed in~\cite{Sun.etal2019b}, the Lipschitz constant is controlled via spectral-normalization~\cite{Sedghi.etal2019}. In both neural nets, the first layer convolves the $H \times W \times C$ input to $H \times W \times 64$ features maps by using 64 filters of size $3 \times 3$.  Here $W$ and $H$ represents the height and width of input image, and $C$ represents the number of image channels. For example, $C = 1$ for gray image and $C = 3$ for color image. The last layer is a single convolutional layer, generating the final output image with $C$ channels by convolving the feature maps with a $3 \times 3 \times 64$ filter. The layers in-between are convolutional layers, each having 64 filters of size $3 \times 3 \times 64$. All layers except for the last one is followed by a rectified linear units (ReLU) layer. Every convolution is performed with a stride = 1, so that the intermediate feature maps share the same spatial size of the input image.


\bibliographystyle{IEEEtran}


\end{document}